\documentclass[sigconf]{aamas} % do not change this line!

\usepackage{booktabs}

\setcopyright{ifaamas} % do not change this line!
\acmDOI{} % do not change this line!
\acmISBN{} % do not change this line!
\acmConference[AAMAS'18]{Proc.\@ of the 17th International Conference on Autonomous Agents and Multiagent Systems (AAMAS 2018)}{July 10--15, 2018}{Stockholm, Sweden}{M.~Dastani, G.~Sukthankar, E.~Andr\'{e}, S.~Koenig (eds.)} % do not change this line!
\acmYear{2018} % do not change this line!
\copyrightyear{2018} % do not change this line!
\acmPrice{} % do not change this line!
\usepackage{amsmath,multicol,amssymb,algorithm,algorithmicx,algpseudocode,graphicx}
\newtheorem{observation}[theorem]{Observation}

\newcommand{\commentout}[1]{}

\begin{document}

\title{Seasonal Goods and Spoiled Milk:\\Pricing for a Limited Shelf-Life} % put your title here!
%\titlenote{Produces the permission block, and copyright information}

% AAMAS: as appropriate, uncomment one subtitle line; check the CFP
%\subtitle{Extended Abstract}
%\subtitle{Industrial Applications Track}
%\subtitle{Socially Interactive Agents Track}
%\subtitle{Blue Sky Ideas Track}
%\subtitle{Robotics Track}
%\subtitle{JAAMAS Track}
%\subtitle{Doctoral Mentoring Program}

%\subtitlenote{The full version of the author's guide is available as \texttt{acmart.pdf} document}

% AAMAS: submissions are anonymous for most tracks
\author{Atiyeh Ashari Ghomi}
\affiliation{%
 \institution{University of Toronto}
 \city{Toronto} 
 \state{Canada} 
}
\email{atiyeh@cs.toronto.edu}

\author{Allan Borodin}
\affiliation{%
 \institution{University of Toronto}
 \city{Toronto} 
 \state{Canada} 
}
\email{bor@cs.toronto.edu}

\author{Omer Lev}
\affiliation{%
 \institution{Ben-Gurion University}
 \city{Beersheba} 
 \state{Israel} 
}
\email{omerlev@bgu.ac.il}

\begin{abstract} % put your abstract here!
We consider a ``price-committment model'' where a single seller announces 
prices for some extended period of time. More specifically, we examine the case of items with a limited shelf-life where storing an item (before consumption) may carry a cost to a buyer (or distributor). For example, eggs, milk, or Groupon coupons have a fixed expiry date, and seasonal goods can suffer a decrease in value. We show how this setting contrasts with recent results by Berbeglia et al \cite{berbeglia:storable} for items with infinite shelf-life.

We prove tight bounds on the seller's profits showing how they relate to the items' shelf-life. We show, counterintuitively, that in our limited shelf-life setting, increasing storage costs can sometimes lead to less profit for the seller which cannot happen when items have unlimited shelf-life. We also provide an algorithm that calculates optimal prices. 
%However, when we consider consumers with limited budgets, we are only able to find a polynomial time optimal pricing method for items with an unlimited (or very long) shelf-life. 
Finally, we examine empirically the relationship between profits and buyer utility as the storage cost and shelf-life duration change, and observe properties, some of which are unique to the limited shelf-life setting.
\end{abstract}

% AAMAS: the ACM CCS are not needed within AAMAS papers
%%
%% The code below should be generated by the tool at
%% http://dl.acm.org/ccs.cfm
%\begin{CCSXML}
%<ccs2012>
% <concept_id>10010520.10010553.10010562</concept_id>
% <concept_desc>Computer systems organization~Embedded systems</concept_desc>
% <concept_id>10010520.10010575.10010755</concept_id>
% <concept_desc>Computer systems organization~Redundancy</concept_desc>
% <concept_significance>300</concept_significance>
% <concept_id>10010520.10010553.10010554</concept_id>
% <concept_desc>Computer systems organization~Robotics</concept_desc>
% <concept_significance>100</concept_significance>
% <concept_id>10003033.10003083.10003095</concept_id>
% <concept_desc>Networks~Network reliability</concept_desc>
% <concept_significance>100</concept_significance>
%</ccs2012> 
%\end{CCSXML}
%
%\ccsdesc[500]{Computer systems organization~Embedded systems}

 % put your semicolon-separated keywords here!

\maketitle

%%%%%%%%%%%%%%%%%%%%%%%%%%%%%%%%%%%%%%%%%%%%%%%%%%%%%%%%%%%%%%%%%%%%%%%%%%%%%%%%%%%%%%%%%%%%%%%%%%%%%%%%%
%% start of main body of paper

\section{Introduction}

The problem of how to allocate resources to different people (or agents) when each of them has a different valuation for a given resource, is one of the most fundamental and well-studied problems in micro-economics. The most common solution has been to set anonymous prices (i.e. identical pricing for every agent) and then agents who value the item above its price buy it, and otherwise they do not.

In the simple multi-unit one-shot scenario setting (i.e. in which buyers with a known valuation for the item make their purchasing decision and leave), finding the optimal price (and hence, the optimal allocation) is a relatively simple optimization problem. However, adding even a small amount of complexity to the scenario makes it significantly harder to solve. Such complications include adding uncertainty about buyers' valuations \cite{Mye81}, multiple vendors \cite{BNL14,BLS16}, and multiple items \cite{MV07,LOBR15} (all with various limitations on the agents' valuation functions). All these problems have spawned intense research to explore their respective areas\footnote{From here on, we use a buyer/seller terminology as it is easier to grasp. However, this applies to many resource allocation problems.}.

Another such issue that leads to an additional complication of the basic problem is adding a temporal dimension to the setting. This means sellers can change their prices over time, and hence buyers can choose to change their buying decisions from day to day, and should they find it worthwhile, to store items over time (introducing the issue of storage cost). Of course, if buyers' valuations remain constant over time, and they wish to buy every day, prices will also remain the same for every day. So the interesting problem involves buyers whose valuation for items change over time. Naturally, the seller wishes to find prices which maximize its profit, while the buyers seek to maximize their own utility (i.e. value at time of consumption $-$ price at time of purchase $-$ storage cost, for each item purchased). This is, fundamentally, a \emph{Stackelberg game}, in which the seller is the ``leader'' setting the prices, while the buyers are the ``followers'' reacting to current and future prices, by pursuing a best response strategy. We examine the outcomes of these games, which are basically the subgame perfect Nash equilibria of the games.

We wish to understand optimal seller pricing (as a function of storage cost and shelf-life duration) and how it impacts the overall utility of the buyers. While there have been several attempts to construct such a model (see Section~\ref{relatedWork}), only recently did Berbeglia et al.~\cite{berbeglia:storable} suggest a model for indivisible items sold over discrete time steps, with buyers being able to store items at a given time for consumption at some later time. The Berbeglia et al.~\cite{berbeglia:storable} analysis is greatly assisted by their insightful result showing that there are seller optimal prices such that buyers will not store items.

We introduce a seemingly small but very natural extension to this model. Instead of discussing items with an unlimited shelf life, we discuss items with a limited consumption date. These can be {\bf perishable} items, like milk, eggs, or fruit, which lose their value after several days, and are no longer fit for consumption. Perishable items aren't only food items; Groupon coupons, for example, also have an expiry date and Amazon Web Services (AWS) server rental periods are another case. A similar family of items that we will discuss are {\bf degradable} items, which diminish their value after several days (though still maintaining some value). Such products can be seasonal or fashion dependent, like clothing items, which significantly lose value once out of season or fashion.

Changing the durability of products may seem small, but it changes the results significantly. The various variables involved in setting prices are effected in a much more direct manner. For example, Berbeglia et al.'s~\cite{berbeglia:storable} simplifying result that buyers will never be required to use storage under optimal pricing no longer holds, and therefore much of their analysis is no longer applicable in our setting. This requires us to explore more directly the effect of storage cost on prices and the resulting impact on buyer decisions. Changing the duration of items allows further examination of the inter-connection between prices and purchasing demand, and how small changes in storage cost or shelf life can cascade into unexpected changes in buyer utility and profits. Sometimes these move in tandem (e.g when a seller lowers the price thereby selling more items to more people, increasing profit and the overall utility of buyers), while in other problem instances this is not necessarily the case.

In this work we examine the issues of profit and buyers' utilities as a function of storage cost and shelf-life both theoretically and empirically. In Section \ref{sec:shelf-life} we show a precise relationship between the shelf-life of an item and the seller's profit, both for perishable and degradable items. In Section \ref{sec:finding-prices} we provide an algorithm for setting prices optimally. Finally, in Section \ref{sec:experiments}, we examine price and social welfare empirically using simulations (with respect to realistic distributions on buyer valuations). We show how a limited shelf-life significantly changes previous results (in Berbeglia et al.~\cite{berbeglia:storable}) on the relation between storage costs and profits, including counterintuitively, that in contrast to infinite shelf life, increasing the cost of storage does not necessarily induce consumers to accept higher prices, and can even reduce the profit of the seller.

\section{Related Work}\label{relatedWork}

While the topic of limited shelf-life has been examined in different settings, for example in \cite{WL12,DHL06}, it has generally treated time and products as completely divisible. In contrast, we examine these issues for discrete time and indivisible goods, as is common for most consumer items. Some work on pricing over time involves agents looking for the cheapest time to buy a single item~\cite{BSS89}, while we examine agents who wish to buy an item each day. The closest work to ours is Berbeglia et al.~\cite{berbeglia:storable}, in which both pre-announced pricing as well as contingent pricing\footnote{That is, ``threat-based'' pricing, in which a seller can announce that if consumers won't buy on day $t$, the price will increase on day $t+1$, otherwise it will stay the same.} were studied. They compared these pricing policies over a finite time period (i.e., there is a known fixed number of days) for an unlimited supply of an indivisible item (e.g., digital goods). They proved that for pre-announced pricing mechanisms with linear (per item unit per day) storage cost and unlimited storage time, there is an optimal set of prices such that for these prices, the consumers do not need to store any goods so as to maximize their utility. They also gave a dynamic program to find the optimal set of prices to maximize the monopolist's revenue.

Slightly further afield, Dasu and Tong \cite{dasu} considered the pricing problem when there is a fixed number of items, goods are perishable and there is a finite time horizon. Beyond their numerical experiments, they showed that if buyers are not strategic, contingent pricing dominates pre-announced pricing, but this result does not hold if consumers are strategic (as they are in our case).

There has been some research about these topics when assuming there are only two (rather than $T$) time periods. Focusing more on pre-announced pricing (as we do), but in a different setting, Correa et al.~\cite{correa} proposed a new pre-announced pricing policy, in which the seller commits to a price menu and dynamically chooses a price in the menu based on available inventory. They considered a limited inventory of an item and different arrival times for consumers. They proved the existence of an equilibrium and they also showed that under certain conditions their pricing policy outperforms contingent and pre-announced pricing policies.

Aviv and Pazgal \cite{aviv} studied the pricing problem in another limited setting, assuming not only 2 time periods, but also assuming consumer arrival times are drawn from a Poisson distribution. They compared a pre-announced pricing policy with a contingent one in which the seller sets the prices based upon the seller's inventory and declining consumer valuations. They argued that the monopolist can increase her revenue up to eight percent in the pre-announced pricing policy compared to contingent pricing.

Our setting is a particular instance of Stackelberg games, on which there has is extensive research, though that has been focusing in the recent past on security games (e.g.,~\cite{KNFBSTJ17,XFCDT16,KFFSTL16}).
\section{Model} 

We study the problem of pricing where a monopolist tries to sell an unlimited supply \footnote{An unlimited supply can be either a digital good (e.g., a newspaper with online subscription), and moreover, in practice we conceptually think that certain items can be produced so as to satisfy any demand. For example, in some countries eggs and milk seem to be in unlimited supply, and in Iceland one can believe that there is an infinite supply of renewable energy.} of a certain product or good at times $1, 2, \ldots, T$. She sets the price for time $i$ to $p_i$, being aware of the valuations of the consumers. She notifies the consumers of the prices for all time periods before purchases commence (i.e., before $t=1$). The number of units of goods sold at time $i$ is $q_i$. The monopolist's objective is to maximize her revenue which is equal to $\sum_{i=1}^Tq_ip_i$.

On the buyer side, we have one or many self-interested rational (i.e., wanting to maximize their utilities) consumers with a value for consuming goods. They can buy on any day and store for consumption on other days. We define the valuation function $v(i, t)$ with domain $\mathbb{N}\times [T]$ for items consumed ``fresh'' (i.e., that have not been stored). Their utility is the value of the items they consume on a given day minus the price they paid for the items and their storage cost. We assume in the case of a tie, the buyer prefers to store the goods as little as possible. Following \cite{berbeglia:storable} we discuss two cases regarding the number of consumers and their demands:

\begin{description}
\item[Multi-buyers]Multiple consumers, each demanding only a single unit of demand. So on day $t$, if we sort the consumers' values for one unit in decreasing order, the consumer's value is $v(i, t)$ which is the $i^{th}$ highest value on day $t$. $v(i, t)$ is non-increasing in $i$.

\item[Single-buyer]One buyer with many units of demand. In this setting, $v(i, t)$ represents the consumer's non-increasing marginal value for the $i^{th}$ unit of goods on day $t$. In other words, $v(i, t)$ is the consumer's value for $i$ units of goods minus the value for $i-1$ units of the goods, so the total value for consuming $i$ units on day $t$ is $\sum_{j=1}^{i}v(j, t)$. We use \cite{berbeglia:storable}'s assumption that there is a cap on the number of items desired by the consumer, i.e., there exists $H\in \mathbb{N}$ such that $v(H, t)=0$ for all $t\in T$.

\end{description}

In the \textit{multi-buyer} setting we have $N$ consumers, and in the \textit{single-buyer} setting, the maximum demand on any day is $N$. The total number of days is $T$. Consumers may have to pay for storing the goods and this cost is defined as storage cost. While \cite{berbeglia:storable} study both linear and concave cost, we only study linear storage cost with cost $c$ per day per unit.

Section \ref{sec:shelf-life} begins the study of items with a limited shelf-life, $d$, after which the item loses all value. We extend this model in 
Section \ref{fractional values} where we consider the pricing problem in the setting where the value of the good when stored becomes a fraction of the initial value. We define a function $r:[T]\rightarrow [0, 1]$ which takes an integer $l$ which is the number of days that a good is going to be stored before consumption and returns a fraction $r(l)$ that specifies the value of the good if consumed $l-1$ days after purchase. The function $r$ is a non-increasing step function. We use this function $r$ to define another valuation function to represent decreasing values. In the \textit{single-buyer} case, let $v'(i, t, l)$ be the value of consuming the $i^{th}$ unit on day $t$, where this unit has been stored $l-1$ days\footnote{In the \textit{multi-buyer} case, $v'(i, t, l)$ is the $i^{th}$ largest value on day $t$ when this unit has been stored $l-1$ days.}. Therefore, we have $v'(i, t, l) = v(i, t)r(l)$. %We will only consider the case where $r = r(\ell)$ for all $\ell \geq d$. 

In \textit{single-buyer} setting, our model is not well-defined yet since the units that are consumed on a particular day can be bought on different days, so the value of the consumed items is not clear. The following immediate observation resolves this definitional issue.

\begin{observation}
In the \textit{single-buyer} setting, if the buyer has chosen some $k$ specific units for day $t$'s consumption that were bought on different days, the order of 
consumption of those $k$ items is to consume those purchased most recently first. So if they are ordered according to number of days they are stored -- $d_{1},\ldots,d_{k}$ -- the consumer's value is $\sum_{i=1}^{k}v'(i,t,d_{i})$.
\end{observation}
%Proof omitted due to space constraints. (Informally, it follows from the marginal value definition of $v(i,t,l)$ and the agent's desire to maximize utility.)
\begin{proof}
To prove this we show that if the consumer buys two units on different days to consume on the same day, the unit which is stored longer must have less value. Hence it is considered the second unit and the other unit is considered the first unit. Let $d_1$ and $d_2$ be the number of days these units are stored and $p_1$ and $p_2$ be the price at which they are bought. We assume $d_1$ is less than $d_2$ (thus, $r(d_{1})\geq r(d_{2}$); therefore, $p_2$ is less than $p_1$, as otherwise, it would be more beneficial to buy unit 2 at price $p_1$ and store it $d_1$ days as well. If the consumer's value for unit 1 is $v_1$ and for unit 2 is $v_2$, we show $v_1$ must be more than $v_2$.

Since the buyer preferred to buy unit 1 when they did ($d_{1}$ days ago; price $p_{1}$), and not buy it $d_{2}$ days ago at price $p_{2}$:

\begin{eqnarray*}
r(d_1)v_1-p_1>&r(d_2)v_1-p_2\\
v_1\cdot(r(d_1)-r(d_2))>&p_1-p_{2}
\end{eqnarray*}

The buyer also preferred to buy unit 2 when it did ($d_{2}$ days ago; price $p_{2}$), and not buy it $d_{1}$ days ago at price $p_{1}$, hence:

\begin{eqnarray*}
r(d_1)v_2-p_1<&r(d_2)v_2-p_2\\
v_2\cdot(r(d_1)-r(d_2))<&p_{1}-p_2
\end{eqnarray*}
Combining these:
\begin{eqnarray*}
v_2\cdot(r(d_1)-r(d_2))<&v_1\cdot(r(d_1)-r(d_2))\\
v_2<&v_1
\end{eqnarray*}
\end{proof}

Practically, in all results the single-buyer and multi-buyer cases are essentially equivalent, and a single proof suffices for both cases.

We note again that this particular setting is an instance of a Stackelberg game, which is defined as a 2-stage game, in which a ``leader'' announces their strategy and the ``followers'' respond to it. A solution to this game is a subgame perfect Nash equilibrium, in which the leader (e.g. a seller) choses the strategy that will maximize their profit assuming that followers (e.g. buyers) will best-respond to it. This is exactly the type of solution we examine here.

\section{$d$-day Shelf-Life}
\label{sec:shelf-life}
In this model, a consumer in the \textit{single-buyer} case or consumers in the \textit{multi-buyer} case can only store the goods for less than $d$ days after which the good is worthless. If $d$ is equal to 1, it means the goods must be consumed on the same day they are bought. Proofs are written for the single-buyer case, but can be easily applied to the multi-buyer one.

\begin{theorem}
The largest possible revenue of the monopolist is a non-increasing function of $d$, and in some cases will be strictly decreasing.
\end{theorem}
%Proof omitted due to space constraints. 
\begin{proof}
Let us assume our buyer can store the goods for $d$ days. We prove that if they can store the goods for $d^\prime=d-1$ days, the monopolist can make as much money as in the $d$-day case.

Assume that the monopolist's best strategy when the buyer can store for $d$ days is $p_1, p_2, \ldots, p_T$. There are two cases regarding the monopolist's best strategy. In the first, the monopolist's prices are such that although the buyer can store the goods for $d$ days, it is not beneficial to do so. Hence, if we reduce the duration to $d^\prime$, the monopolist can use the same strategy making the same amount of money. 

The second case is when it is beneficial for the buyer to store some goods for $d$ days. In this case, we describe a new strategy $p^\prime_1, p^\prime_2, \ldots, p^\prime_T$ by which the monopolist makes at least the same amount of money. Let us assume day $s$ is the last day that our buyer is going to store one unit of the goods to consume $d$ days later; that is, to consume on day $s+d-1$. As noted before, in the case of a tie the buyer prefers to store the goods as little as possible. Therefore, $p_t>p_s+(t-s)c$ for all $s<t\leq s+d-1$ since otherwise, the buyer would be better off buying the extra units of goods on day $t$ to consume on day $s+d-1$, instead of day $s$.
The buyer is not going to store for $d$ days on day $s+1$, since $s$ is the last day that the item is going to be stored for $d$ days. Moreover, now items will be stored for less than $d$ days on day $s+1$ since $p_{s+1}>p_s+c$ as stated above. Therefore, $q_{s+1}=0$. When the buyer can only store the goods for $d^\prime$ days, if we set $p^\prime_{s+1}=p_s+c$, the buyer's behavior, in terms of purchase, for the days before $s$ do not change at all (since it is not possible to buy on day $s$ or $s+1$ anyway to consume on previous days, so later prices do not need to be taken into account). The behavior for the days after $s+1$ also does not change because prices did not change. The only changes are at times $s$ and $s+1$ when the buyer bought on day $s$ and stored for future. Under the new prices, the item can be bought on day $s+1$ instead since $p^\prime_{s+1}$ is equivalent to $p_s$ for the buyer. The amount of the goods does not change since from buyer's perspective their utility has not changed either. Hence, overall, cost of buying and storing the goods has not changed from the $d$-day case.
\end{proof}

\begin{corollary}
As a result, the monopolist makes the most money when the goods must be consumed on the day they are purchased (i.e., $d=1$). 
\end{corollary}

We begin exploring our limited shelf-life problem by noting that a significantly useful and simplifying result from \cite{berbeglia:storable} (Theorem 3.1) does \emph{not} hold in our case. In their setting (i.e., for items with infinite shelf-life), optimal pricing results in buyers not using storage at all; that is, buyers consume items on the days they buy them.

\begin{theorem}
\label{need storage}
There are settings where the monopolist will make less money as the storage cost is increased. Additionally, in this setting with limited shelf-life, the best strategy for the monopolist sometimes makes consumers store the items.
\end{theorem}
\begin{proof}
We give an example of the situation where the monopolist makes less money when the storage cost increases, and where the consumer will always use the storage. In this example, let $d = 2$ and the cost of storage $0$. As shown in Table~\ref{table:counter example}, we assume that there is a single consumer whose values for one unit of the goods on days 1, 2 and 3 are 1, 1 and 1000, respectively and for additional units is 0. Because the largest value is equal to 1000, the best price for days 2 and 3 must be 1000 to ensure the large payment on day 3. For the first day, best price is 1.

When the storage cost is 0, the consumer will buy two units on the first day, to consume on days 1 and 2 -- unlike the infinite shelf-life case, storage \emph{must} be used for optimal case. However, when the storage cost is 2, the consumer will only buy one unit on the first day. Hence, the monopolist's profit is reduced with the increase in costs.
\begin{table}
\begin{center}
\begin{tabular}{|c|c c c|}
\hline
 Day: & {\bf1} & {\bf2} & {\bf3}\\
 \hline 
 Consumer valuation & 1&1&1000\\
 Optimal price (all cases) &1 &1000&1000\\
Seller profit when storage cost is 0 & 2 & 0 & 1000\\ 
Seller profit when storage cost is 2 & 1 & 0 & 1000\\
 \hline
\end{tabular}
\caption{Example in Theorem~\ref{need storage}}\label{table:counter example}
\end{center}
\end{table}
\normalsize
\end{proof}

Throughout this Section let $M$ be the maximum amount of money that the monopolist makes when goods are always consumed on the day of purchase and cannot be stored (i.e. $d = 1$).

\begin{theorem}\label{lowerbound d-day}
When goods can be stored for $d$ days, the monopolist makes at least $\frac{M}{d}$ amount of money for any \footnote{In fact, this theorem holds for an arbitrary weakly monotonic cost function but we are only considering linear cost functions in this work.} linear storage cost function $c$.
\end{theorem}
\begin{proof}
We set prices so that the monopolist will make $\frac{M}{d}$ amount of money for any linear cost function and any set of consumer (or consumers, in the multi-buyer setting) values. Let $p_1, \ldots, p_T$ be the optimal prices when $d = 1$ and let $q_1, \ldots, q_T$ be the number of units purchased on each day. Now, to set the prices, consider $d$ different pricing options. In each case $t$ ($0< t\leq d$), the goods' price, on day $sd+t$ for all $0\leq s<{\lfloor}{\frac{T}{d}}{\rfloor}$ is equal to $p_{sd+t}$ and the goods' prices for other days are very high (effectively, $\infty$). The buyer will not store the goods for day $sd+t$ because the prices on days $(s-1)d+t+1, (s-1)d+t+2, \ldots, sd+t-1$ are large numbers and greater than $p_{sd+t}$. %Now $v(q_{sd+t}, sd+t, 1)\geq p_{sd+t}$. Recall $v(i, t, l)$ is the $i^{th}$ largest value on day $t$ when this unit has been stored $l-1$ days after purchase. Consumers do not store for later consumption on day $sd+t$; 
Therefore, the amount bought on day $sd+t$ will be at least $q_{sd+t}$. Hence, summing over the $d$ different pricing schemes, the sum of the revenue in these $d$ cases is more than or equal to $M$, so there is at least one of them for which the revenue for the $d$-day case is more than or equal to $\frac{M}{d}$.
\end{proof}

That theorem showed a lower bound for the seller's profit. We now show this bound is tight: 

\begin{theorem}\label{upperbound d-day}
For any $\epsilon>0$, there is a setting in which the monopolist's maximal revenue is less than $(1+\epsilon)\frac{M}{d}$.
\end{theorem}
\begin{proof}
Choose a natural number $a$ such that $\frac{1}{a-1}<\epsilon$. At first, assume $T=d$, then later, we will complete the proof for $T=kd$ for $k\in \mathbb{N}$.
%Assume $T=d$ and $c=0$. 
For simplicity, we will assume that $c = 0$\footnote{For $c>0$, by increasing $b$ to an arbitrarily high value, we can make values large enough, and the difference between each day significant enough, so the behavior is practically as if $c=0$.}. We now define $b$ as $\prod_{t=1}^d(a^{d-t+1}-1)$. For each day $t$, $1\leq t\leq d$, the buyer wishes to purchase $a^{d-t}$ items, each with a value of to $\frac{b\cdot(a-1)}{a^{d-t+1}-1}$. Any additional item has a value of $0$.
% in the \textit{muli-buyer} case or the consumer has marginal value equal to $\frac{b\cdot(a-1)}{a^{d-t+1}-1}$ for first $a^{d-t}$ units and 0 for others in the \textit{single-buyer} case. We will only discuss the \textit{multi-buyer} case since the other case is similar. 
%Now, define $b=\prod_{t=1}^d(a^{d-t+1}-1)$.

For the monopolist, it is beneficial to set the price to one of the valuations of the buyer, since otherwise, it can increase the price without losing any purchase, so prices are of the set $\{\frac{b\cdot(a-1)}{a^{d-t+1}-1}: 1\leq t\leq d\}$. Increasing prices as time goes on does not increase the revenue, since the buyer can buy when prices are lower and store for later, as storage cost is 0. Since values are going up, reducing the price does not increase the revenue either. Therefore, the monopolist just sets a fixed price for all days, which, as noted, should be equal to some item's value. So the monopolist's revenue equals $\frac{b\cdot(a-1)}{a^{d-t+1}-1}$ (for some $0<t\leq d$) times the number of items with value more than or equal to the price. I.e., $\frac{b\cdot(a-1)}{a^{d-t+1}-1}\cdot \sum_{t=0}^{d-t}a^t=\frac{b\cdot(a-1)}{a^{d-t+1}-1}\cdot\frac{a^{d-t+1}-1}{a-1}=b$. In comparison, $M=\sum_{t=1}^d\frac{b\cdot(a-1)}{a^{d-t+1}-1}\cdot a^{d-t}$. Monopolist revenue in the $d$-day case compared to $M$ is:\\
\small
$$
\frac{b}{\sum_{t=1}^d\frac{b\cdot(a-1)}{a^{d-t+1}-1}\cdot a^{d-t}}=\frac{1}{\sum_{t=1}^d\frac{(a-1)}{a^{d-t+1}-1}\cdot a^{d-t}}$$
$$\frac{1}{\sum_{t=1}^d\frac{a^{d-t+1}-a^{d-t}}{a^{d-t+1}-1}} < \frac{1}{\sum_{t=1}^d\frac{a^{d-t+1}-a^{d-t}}{a^{d-t+1}}}$$
$$=\frac{1}{\sum_{t=1}^d1-\frac{1}{a}}=\frac{1}{d-\frac{d}{a}}=\frac{1}{d(1-\frac{1}{a})}=\frac{1}{d(\frac{a-1}{a})}$$
$$\frac{(\frac{a}{a-1})}{d}=\frac{(1+\frac{1}{a-1})}{d}<\frac{(1+\epsilon)}{d}$$
\normalsize

So far we showed that if $T=d$, the total revenue is less than $(1+\epsilon)\frac{M}{d}$. In a more general case, we set $T=kd$ which means we have $k$ blocks of length $d$. On day $t$, $1\leq t\leq d$ in block $i$, $0\leq i< k$, the buyer wants to buy $a^{d-t}$ items, each with a value of $\frac{b\cdot(a-1)}{a^{d-t+1}-1}b^{k-i}$ (and additional items are valued at $0$). Therefore, the optimal prices in each block $i$ are also multiplied by $b^{k-i}$. Since in each block compared to its previous block, prices are lower, the buyer would not store any goods from the previous block. We define $M_i$ for $0\leq i<k$ as the maximum achievable revenue for block $i$ when there is no storage. As proved, in each block $i$, the maximum revenue is less $(1+\epsilon)\frac{M_i}{d}$. Therefore, in general, the maximum revenue is also less than $(1+\epsilon)\frac{M}{d}$.
\end{proof}

\subsection{$d$-Day Fractional Value}\label{fractional values}
Generalizing our shelf-life results from the previous Section, instead of assuming that after $d$ days the goods' values drops to $0$, we assume that after $d$ days the goods' value drops to a fraction $r$ ($0\leq r < 1$) of its value when bought. In other words, for $1\leq i\leq N$ and $1\leq t\leq T$, $v'(i, t, l) = v'(i, t, 1)$ for $l\leq d$ and $v'(i, t, l) = r \cdot v'(i, t, 1)$ for $l>d$.

The results are, to a large extent, a generalization of the $r=0$ case.
%In this model, in the \textit{single-buyer} case, the units that are consumed on a day can be bought on different days and some of them can be stored for more than $d$ days.

%The following theorem holds in both \textit{single-buyer} and \textit{multi-buyer} settings.

\begin{theorem}\label{lowerbound fractional values}
In $d$-day storage with fractional value model, the monopolist makes at least $\frac{1-r}{d}M$ amount of money.
\end{theorem}

\begin{theorem}\label{upperbound fractional values}
For all small $\epsilon>0$ there is a setting with $d$-day storage with fractional value model in which the monopolist's maximum revenue is less than $(\frac{1-r}{d}+\epsilon) M$.
\end{theorem}
The proofs for these two theorems are somewhat more complicated than for the corresponding Theorems~\ref{lowerbound d-day} and \ref{upperbound d-day}, but as they follow a similar structure we omit the proofs due to space constraints.

\section{Finding Optimal Prices in $d$-Day Shelf-Life}
\label{sec:finding-prices}

The monopolist's goal is to maximize revenue, while consumers aim to maximize utilities. Consequently, when the prices are announced by the monopolist, the consumers seek the best strategy for them, which manifests itself in the number of units bought each day and the number of units consumed each day.

In the \textit{multi-buyer} setting, each consumer starts from day $T$ and works backward, trying to find the best day to purchase the unit that will be consumed on day $T$. Then the consumer proceeds to day $T-1$, repeating the process, and then moves on to day $T-2$ and onwards. We have $N$ consumers and each of them finds their best strategy in time $T^2$, so the running time of this algorithm is $T^2N$. This same algorithm also works in \textit{single-buyer} setting. In \textit{single-buyer} setting, on each day, for each marginal value greater than 0, the consumer finds the best day to buy a unit to maximize the utility separately. Therefore, the running time is $T^2$ for each unit, multiplied by the maximal number of units which is $N$. This algorithm works for any storage model ($d$-day storage, more than $d$-day storage with fractional values and multi-step value decrease model which we define later in this paper) and many storage cost functions beyond the linear we mainly address here.

On the other hand, from the monopolist's point of view, finding the best prices is not as easy as finding the best strategy for consumers. In this Section, we deal with finding the best strategy for the monopolist. We present an algorithm, exponential in $d$, which finds the best prices in the $d$-day storage model in both the \textit{single-buyer} and \textit{multi-buyer} settings.

The next theorem is similar to Theorem 3.2 in \cite{berbeglia:storable}. However, as in our model the best pricing sequence may require storage, the proof and the theorem are not the same.

\begin{theorem}\label{possible prices}
There exists an optimal pre-announced pricing sequence $p_1, p_2, \ldots, p_T$ such that for each $t$, we have $p_t=v(i, s, 1)+c(t-s)$ for some $1\leq s\leq T$ and some $0\leq i\leq N$.
\end{theorem}
%Proof omitted due to space constraints.
\begin{proof}
The difference between this theorem and theorem 3.2 in \cite{berbeglia:storable} is that here, $1\leq s\leq T$ whereas in \cite{berbeglia:storable} $1\leq s\leq t$. Notice that here, $t-s$ can be negative, in which case the price for day $t$ is equal to $v(i, s, 1)$ minus storage cost. To prove this theorem, we set $v(0, t, 1)=L$ for all $t\in [T]$ where $L$ is a large number. We make this assumption because on days that nothing is sold, we set prices to $L$.

Let us assume $\{p_1, p_2, \ldots, p_T\}$ is the set of optimal prices and if there are several optimal sets of prices, choose one set arbitrarily. Take the smallest $t$ such that $p_t\neq v(i, s, 1)+(t-s)\cdot c$ for any $1\leq s\leq T$ and any $1\leq i\leq N$. If $q_t=0$, then set $p_t^\prime=v(0, t, 1)=L$. Clearly, the consumers still do not buy anything on day $t$ because the price on this day is a large number and consumers' behaviour on other days does not change either. Therefore, the monopolist did not lose any money by this change. If $q_t>0$, set $p_t^\prime=\min\{v(j,s, 1)+(t-s)\cdot c: 1\leq s\leq T; 1\leq j\leq N; v(j,s, 1)+(t-s)\cdot c>p_t\}$. If $q_t>0$, then $p'_t$ is well-defined because $v(j', s', 1)+(t-s')c\geq p_t$ for some $1\leq j'\leq N$ and $1\leq s'\leq T$ in order to have $q_t>0$. Now we are going to prove why this new set of prices is more profitable. On any day $t'<t$, the consumers will buy those units of goods they bought previously since the prices did not change on these days and future prices are either increasing or staying the same. On any day $t'>t$, again the consumers will buy those units of goods they bought previously since additional purchases on day $t$ for future consumption are not beneficial.

We only need to study what happens on day $t$. With our optimal prices, consumers bought $q_t$ units on day $t$. These $q_t$ units were consumed on different days, among all of these units, consider the one which had the least value for consumers. Let us say this value is the value of the $k^{th}$ unit on day $t^\prime$, $v(k, t^{\prime}, 1)$. So we have $p_t + c\cdot(t^\prime - t)\leq v(k, t^\prime, 1)$, but because $p_t\neq v(i, s, 1)+(t-s)\cdot c$ for any $1\leq s\leq T$ and any $1\leq i \leq N$, we have $p_t < v(k, t^\prime, 1)$. We know $v(k, t^\prime, 1)+c\cdot (t-t')$ belongs to $\{v(j,s, 1)+(t-s)\cdot c: 1\leq s\leq T; 0\leq j\leq N; v(j,s, 1)+(t-s)\cdot c>p_t\}$, so the new price $p_t^\prime$ which is the minimum value of the set is less than or equal to $v(k, t^\prime, 1)+c\cdot (t-t')$. Therefore, the consumer still affords to buy those units of goods, but they may prefer to buy them on other days rather than day $t$. These other days cannot be any day before day $t$ because we assumed that day $t$ was the first day that $p_t\neq v(i, s, 1)+(t-s)\cdot c$ for any $1\leq s\leq T$ and any $1\leq i\leq N$. Thus, prices on all previous days are in that form and $p_t^\prime$ is minimum value of the prices in that form. As a result, the consumers do not prefer to buy those units of goods on earlier days. It is possible that consumers buy those units on days after day $t$. In this case, the amount of money that consumers are paying is more than or equal to previous amount since previously, they preferred to buy on day $t$. Besides, they are storing for fewer days; therefore, the amount of money the monopolist makes is more than or equal to the previous amount. In conclusion, the monopolist does not lose any money by this change.
\end{proof}

In order to give a dynamic program to find the optimal prices in the $d$-day storage model, we need some definitions. First, we define $C_t$ which is the set of possible prices considering only future prices at time $t$ and $C_t^\prime$ which is the set of all possible prices. $C_t=\{v(j, s, 1)+(t-s)\cdot c|t\leq s \leq T, 0\leq j\leq N\}$ and $C^\prime_t=\{v(j, s, 1)+(t-s)\cdot c|1 \leq s \leq T, 0\leq j\leq N\}$.
% Now, we define two different prices; first,
Next, we define price $p^\prime_t$ to be the price the monopolist sells the goods that consumers are going to consume on day $t$; it can be sold on any day up to and including day $t$. Now we define cost $p''_t$ to be the total cost that consumers have paid for the goods to be consumed on day $t$, i.e., purchase price + cost of storage. %Since in the optimal pricing, consumers may need to store, $p'_t$ and $p''_t$ are not necessarily equal. 

We need to define additional functions: $p''_t(x_1, x_2, \ldots, x_d)$ takes prices $x_1, \ldots, x_d$ which are prices on days $t-d+1, t-d+2, \ldots, t$ and returns the lowest cost for the buyer (that is, including storage cost) to buy item for day $t$. We also use $argmin_t p''_t(x_1, \ldots, x_d)$ to return the index of the day with the lowest cost considering storage cost for day $t$. %So we have $argmin_tp''_t(x_1, x_2, \ldots, x_d)=arg \min_i\{x_i+(d-i)\cdot c|1\leq i\leq d\}$; , $p''_t(x_1, x_2, \ldots, x_d)=\min\{x_i+(d-i)\cdot c|1\leq i\leq d\}$ and 
Note that $p'_t=x_{argmin_tp''_t(x_1, x_2, \ldots, x_d)}$. Finally, we define $q'_t(x_1, x_2, \ldots, x_d)$ which is the number of units of goods which consumers will buy to consume on day $t$, it can be purchased on any day up to and including day $t$, $q'_t(x_1, x_2, \ldots, x_d)=|\{j\geq 1: v_{j, t}\geq p''_t(x_1, x_2, \ldots, x_d)\}|$

\begin{theorem}
The dynamic program (algorithm \ref{d-day algorithm}) finds the optimal prices.
\end{theorem}
\begin{algorithm}
\begin{flushleft}
\small
\caption{Optimal Prices in $d$-day Model}\label{d-day algorithm}
\begin{algorithmic}[1]

\State $R(T+1, x_1, x_2, \ldots, x_{d-1}) \gets 0$ for all $x_1, \ldots, x_{d-1}\in (\bigcup\limits_{t=1}^{T}C'_t)^{d-1}$
\For{$t=T\rightarrow d$}
	\ForAll {$x_1, x_2, \ldots, x_{d-1}\in C^\prime_{t-d+1}\times C^\prime_{t-d+2}\times \ldots, C^\prime_{t-1}$}
		\State{$R(t, x_1, \ldots, x_{d-1}) \gets \max\{q'_t(x_1, \ldots, x_{d-1}, x_d)p_t^\prime(x_1, \ldots, x_{d-1}, x_d)+R(t+1, x_2, x_3, \ldots, x_d) : x_d\in C_t\cup\{x_i+(d-i)c, 1\leq i\leq d\}\}$}
		\State{$S(t, x_1, x_2, \ldots, x_{d-1})\gets arg\max_{x_d}\{q'_t(x_1, \ldots, x_{d-1}, x_d)p_t^\prime(x_1, \ldots, x_{d-1}, x_d)+R(t+1, x_2, x_3, \ldots, x_d) : x_d\in C_t\cup\{x_i+(d-i)c, 1\leq i\leq d\}\}$}
	\EndFor
\EndFor
\State{$x^\star_1, x^\star_2, \ldots, x^\star_{d-1}\gets arg\max_{x_1, \ldots, x_{d-1}}\{R(d, x_1, x_2, \ldots, x_{d-1}): x_1, x_2, \ldots, x_{d-1} \in C^\prime_1\times C^\prime_2\times \ldots \times C^\prime_{d-1}\}$}	
\For{$t=d\rightarrow T$}
	\State{$x^\star_t\gets S(t, x^\star_{t-d+1}, x^\star_{t-d+2}, \ldots, x^\star_{t-1})$}
\EndFor

\Return $x^\star_1, x^\star_2, \ldots, x^\star_T$
\end{algorithmic}
\end{flushleft}
\end{algorithm}
\normalsize

\begin{proof}[Sketch of proof]
$R(t, x_1, \ldots, x_{d-1})$ computes the optimal revenue that the monopolist can earn from day $t$ to day $T$ given that $x_1, \ldots, x_{d-1}$ are prices on $d-1$ previous days. This is done by backwards induction. First, we have $R(T+1, x_1, x_2, \ldots, x_{d-1})$ to zero for any $x_1, \ldots, x_{d-1}$ and second, in the for loop when $t=T$, it finds $R(T, x_1, \ldots, x_{d-1})$ for any given $x_1, \ldots, x_{d-1}$ by going through all possible prices using $x_d$ variable for day $T$. The algorithm checks how many units consumer will buy to consume on day $T$ with given prices $x_1, \ldots, x_{d}$ by calculating $q'_t(x_1, \ldots, x_{d-1}, x_d)$ and then computes how much the monopolist will earn per unit by calculating $p_t^\prime(x_1, \ldots, x_{d-1}, x_d)$; thus, taking the $\max$ of their multiplication is the maximum total amount of money that the monopolist makes for day $T$ consumption. We keep the knowledge of the price we chose using $S$.

For the induction step, for any day $k$ and any possible $x_1, \ldots, x_{d-1}$, we assume that we computed $R(t, x_1, \ldots, x_{d-1})$ for all days $k<t\leq T$, then given $x_1, \ldots, x_{d-1}$, we compute the obtainable revenue from day $t$ to $T$ for all possible $x_d$. Then we find the maximum of these values as $R(k, x_1, \ldots, x_{d-1})$.

Finally, we compute the best prices for first $d-1$ days, by checking all possible prices for those days and computing the maximum revenue. Therefore, by using first $d-1$ prices and $S(t, x_1, x_2, \ldots, x_{d-1})$ we can find the optimal prices for all days.
\end{proof}

\begin{observation}
The running time of the dynamic program is $O((N T)^d d T)$.
\end{observation}
\begin{figure*}[h]
 \begin{multicols}{2}

\begin{center}
\includegraphics[width=\textwidth]{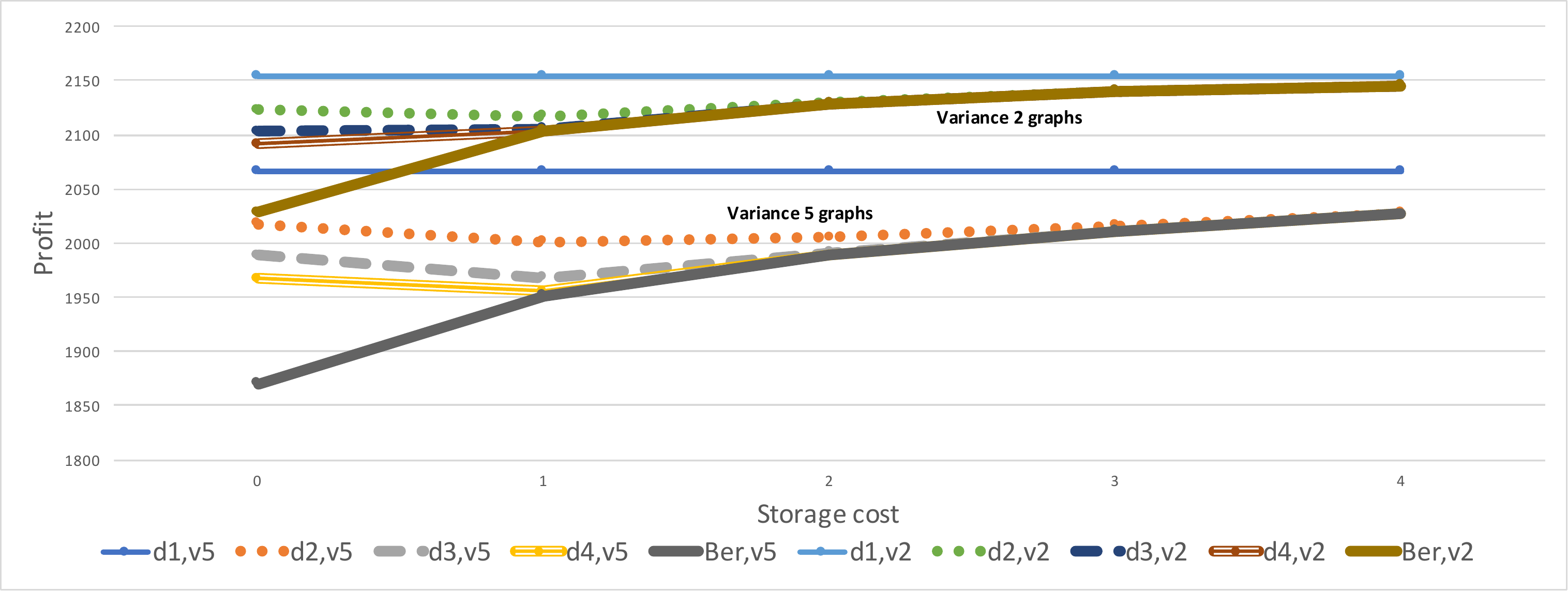}
\end{center}
\end{multicols}

\caption{Effects on profit when changing storage cost, with different graphs according to shelf-life duration and variance of distribution from which utility valuations were taken (``Ber''' indicates \cite{berbeglia:storable}'s model of infinite shelf-life).}
\label{profitCost}
\end{figure*}

\section{Empirical Examination of $d$-day Shelf-Life}
\label{sec:experiments}

We designed a set of simulations so as to more carefully examine the connections between prices, buyers' utilities, storage costs and shelf life. As we wish to understand these relations in realistic settings, we chose buyer valuation functions corresponding to a consumer product. We did this by first choosing for each buyer $i$ their ``base value'' $v_{i}$ for one unit of an item (e.g., how much does one like apples), using a normal distribution with a fairly large variance (we used one with mean 30 and variance 10). However, if each buyer's valuation was fixed the pricing problem would simply be a matter of finding the optimal price for a single day. Hence, as in real life, one's daily valuation is close to, but not exactly, their ``base value'' but not exactly it (e.g., some days one can be busier, without time for a snack). Therefore, we specify a buyers' valuation distribution as a normal distribution with its mean being its base value, $v_{i}$, and its variance being either 5 or 2 (we chose to see the different behavior when valuations change more or less significantly each day). 
We ran this experiment with $N = 5$ buyers and time-horizon $T=20$.
%\footnote{
%(Since there were no significant changes in the results when varying $T$, we only show the results for $T=20$.).

What is the impact of rising storage costs on prices, profit and utility? The seller can respond to rising storage cost by increasing or decreasing prices (or not respond at all) so as to obtain optimal revenue. Increasing prices can benefit profit in an obvious way if one does not drive out too many buyers on any given day. Decreasing prices can result in more profit by allowing more buyers to make a purchase if the increased participation offsets the lower prices. We recall the critical observation in Berbeglia et al.~ \cite{berbeglia:storable} that there is no need for storage with optimal prices when there is unlimited shelf-life. Hence, it follows that {\bf increasing storage costs cannot decrease profit in the unlimited shelf-life model}, since any buyer who did not store before (even at cost 0) will surely not want to store at a higher storage cost. This allows the seller to increase or decrease prices so as to achieve optimal revenue by determining the tradeoff between the increase in price per item sold to buyers who continue to buy and the loss due to buyers who will not buy on a given day. 

However, as shown by Theorem~\ref{need storage} and illustrated in Table~\ref{table:counter example}, {\bf in the limited shelf-life model, storage is sometimes necessary, and profits can actually decrease when storage costs rise}. To what extent does this 
happen in the reasonably realistic scenario given by our distribution on buyer values? Clearly, the smaller the variation in each buyers valuation, the closer we are to 
simple identical pricing for every day without any anomalies and conversely, we may expect that pricing becomes more subtle as the variation increases. Similarly, the longer the shelf-life duration $d$, the closer we are the unlimited shelf-life model. 

 \begin{figure*}[h]
 \begin{multicols}{2}
\begin{center}
\includegraphics[width=\textwidth]{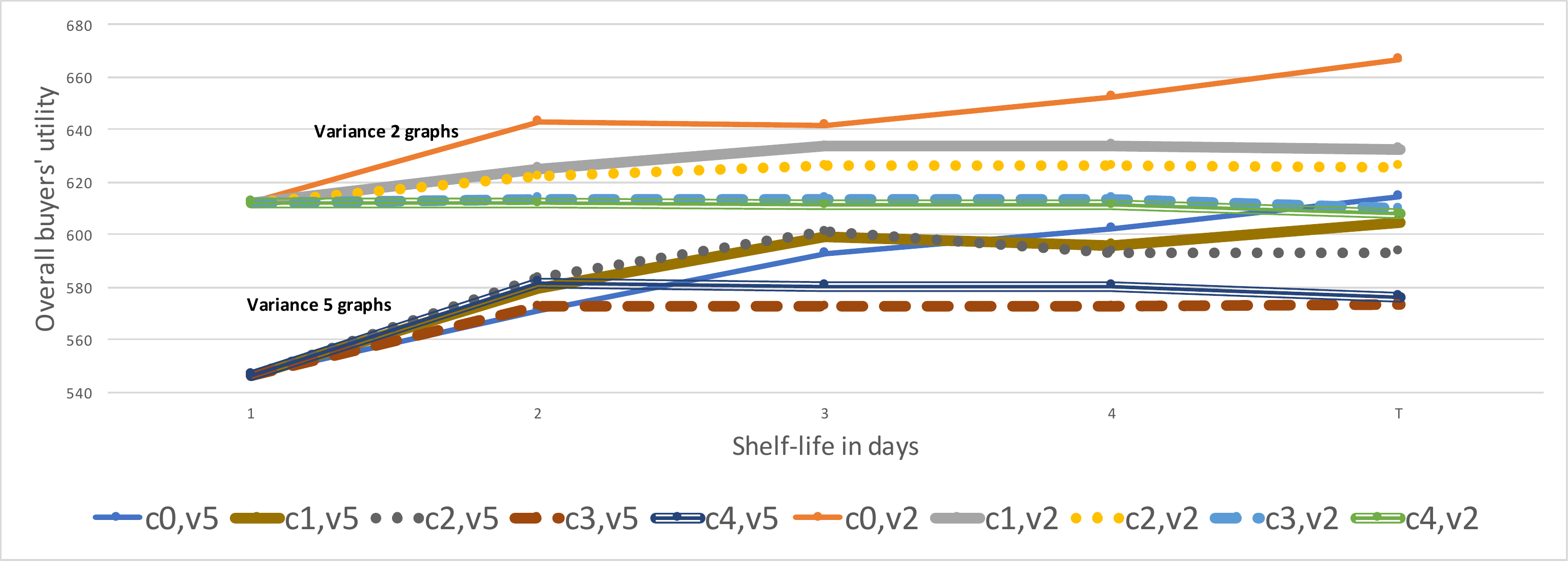}
\end{center}
\end{multicols}
\caption{Effects on overall buyers' utility when shelf-life duration increases (until it is $T$, which are the results for \cite{berbeglia:storable}), with different storage costs and variance of distribution from which utility valuations were taken.}
\label{SWduration}
\end{figure*}

Figure \ref{profitCost} demonstrates the impact of rising storage costs on the sellers optimal revenue. Note the various graphs with variance 2 are always above their equivalent with variance 5. We observe that the two curves (for variance 2 and 5) for unlimited shelf-life are indeed monotonically increasing with cost. In contrast, even for small variance, the curves for limited duration are not monotonic and that this phenomena is more accentuated with higher variance although the curves do become monotonic as the duration increases. 

%As is standard, we define the social welfare to be the sum of the values for items that are sold. (Equivalently, it is the combined utility of thr seller and all the buyers.) 
Figure~\ref{SWCost} considers the overall utility of the buyers (i.e. 
the sum of utilities for each item unit sold) as a function of storage cost. We note that in the effects of rising costs described above, only the one lowering prices has the possibility of increasing the overall utility while increasing profitability. In order for this to happen, there needs to be more than a small difference between the valuations in different days, and the longer the shelf-life, the larger is the seller's concern that one could buy the item when it's cheap and save it. Hence, the longer the shelf-life, the storage cost needs to be higher, so it would not be beneficial for a buyer to buy and store. For the higher variance this subtle interplay is apparent in the figure. Considering both Figures~\ref{profitCost} and \ref{SWCost}, we can see that for duration $d=4$, increasing the storage cost from $0$ to $1$ illustrates that both the profit and overall utility can decrease. Note also that higher shelf life duration will allow the buyer to store more often but the cost of storage tends to lower the overall utility. 

 \begin{figure}
\begin{center}
\includegraphics[width=\columnwidth]{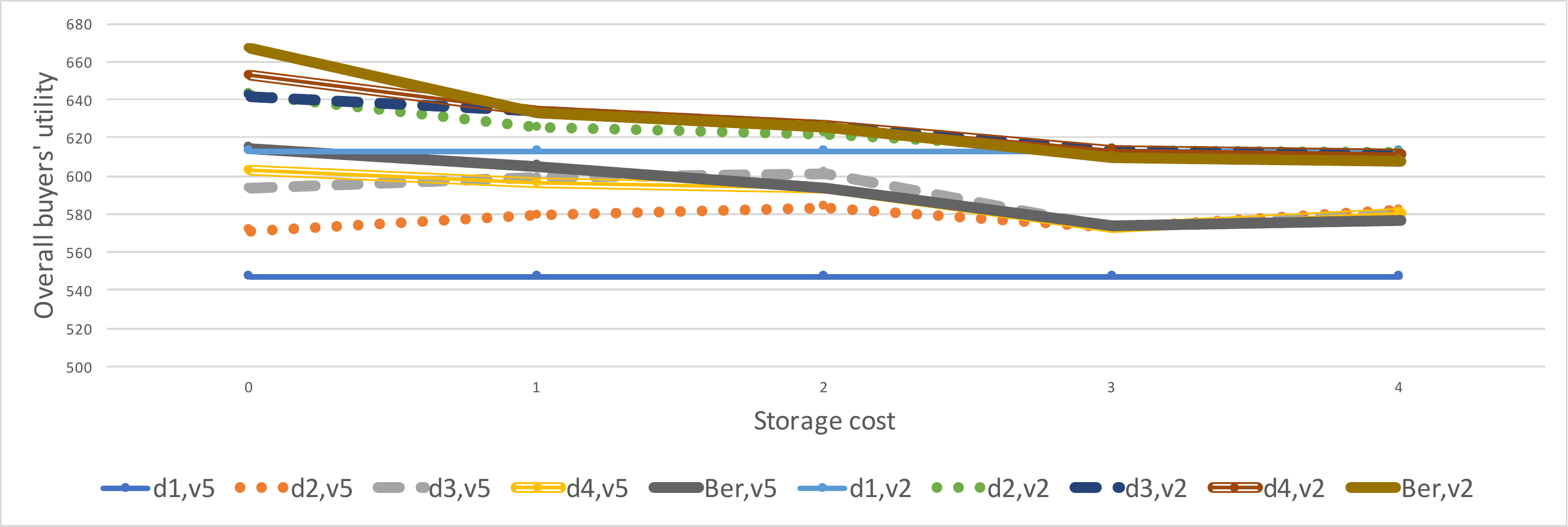}
\end{center}
\caption{Effects on overall buyers' utility when changing storage cost, with different graphs according to shelf-life duration and variance of distribution from which utility valuations were taken (``Ber''' indicates \cite{berbeglia:storable}'s model of infinite shelf-life).}
\label{SWCost}
\end{figure}

The differing variation of the buyers' utility from day to day has, as is to be expected, a significant effect on the observed phenomenon. When the variation is smaller, it can be approximated by the case where the valuations are the same, which are far easier to analyze (since prices stay fixed throughout). Indeed, as can be seen in Figure~\ref{SWduration}, for the low variance case, a higher storage cost goes hand in hand with lower buyers' utility, an effect which becomes more accentuated with the shelf-life duration. However, when the variance is higher, this clean and orderly structure disappears. As we have observed before, unlike \cite{berbeglia:storable}'s model, higher storage costs are not necessarily linked with lower utility, and this effect is clearer when the shelf-life is shorter; the longer it gets, the closer it resembles \cite{berbeglia:storable}'s model, in which the shelf-life is $T$. In particular for shorter shelf-life, the interaction between prices and storage costs is quite intricate, resulting in increased profits for the seller, for whom the storage costs are a guarantee that a lower price on a certain day would not ``propagate'' to future days.

\section{Conclusion and Future Directions}
We first studied the $d$-day shelf-life pricing problem (when items perish in $d$ days), and then we extended the model so that an item retains a fraction of its value after $d$-days. We proved tight bounds on the seller's profits in these models, which show the profit decreases linearly as the shelf life grows. For the $d$-day shelf-life model we gave an algorithm (polynomial time in $N$ and $T$ but exponential in the shelf life $d$) to calculate optimal prices. One immediate question is whether or not this exponential 
dependence on the shelf -life $d$ is necessary. While for many perishable food items one would expect $d$ to be relatively small (i.e., relative to the overall time frame $T$ for which decisions are being made), but in other applications, $d$ might be quite large. 
%We then considered the additional complexity of adding a budget constraint for consumers, and gave an algorithm to find optimal prices in the case of non-perishable goods. 

As noted in the introduction, optimal pricing calculations are, de-facto, finding an allocation mechanism that can be applied in various settings of limited resources, and our time-sensitive setting has applications beyond rotten eggs and out-of-fashion clothing items. 
For example, cloud services -- the usage of which is growing significantly -- are commonly priced so that users pay for a set of resources they can only use for a limited time, which is exactly a limited shelf-life product.

%And while this paper takes the first steps in such an analysis, there 

There are many ways to continue and expand this line of research. Our setting did not include the presence of consumer budgets, that is, an overall limit on the expenditure buyers can afford throughout the whole period $T$. This is an issue not only in our setting, but also in Berbeglia et al.~\cite{berbeglia:storable}. While we have some preliminary results in this regard 
(namely a quadratic programming algorithm), 
the presence of budgets leads to a substantially more complex pricing problem, as was shown in different settings to which budgets were added (e.g.,~\cite{BLS16}). Naturally, budget considerations will come into play even more significantly when extending the model to consider prices for multiple distinct items with one or multiple sellers. If there is no budget then item pricing may be considered as separate sales; but with budgets, to what extent would item prices be related? 

%whether or not there is a polynomial time (in $N,T$ and $d$) algorithm for the $d$-day shelf-life model. Another algorithmic question is whether there is a polynomial time algorithms to calculate the optimal prices for the $d$-day shelf-life item when buyers have a budget. (We currently have a quadratic programming algorithm for this problem which was omitted due to space constraints.) 
%If any case, finding efficient approximation algorithms is an interesting issue. 
%Another interesting direction is to consider prices for multiple distinct items. 
%Naturally, in cases where optimal solutions cannot be calculated in polynomial time, expanding and diversifying our set of simulations may assist in better understanding the problem.

A further extension of the $d$-day fractional model is to allow an item's value to decrease gradually, so that after $d_{k}$ days (for $k = 1, \ldots, t$), the value of the item decreases to a fraction $r_{k}$ of its initial value 
%(for some values of $d_{i},r_{i}$) . 
until it eventually (after some $d_{t+1}$ days) loses all value. An obvious (but mistaken) approach to this would be to assume Theorems~\ref{lowerbound fractional values} and~\ref{upperbound fractional values} can easily be nested. This does not work, since buyers can always buy a completely new item, and while we hypothesize the outcome will be a linear relation between $r_{k}$ and the profits, it requires a different approach than the one used here. Another fundamental change is moving to an adaptive pricing model; namely, instead of pre-announcing prices, how will the market behave when the seller changes prices dynamically as discussed in Berbeglia et al.~\cite{berbeglia:storable}.

Finally, an additional topic of consideration -- not only for our model, but for~\cite{berbeglia:storable} and others as well -- is one of information. Our scenario assumes a full information setting where the seller knows the valuations 
of buyers for each day. What should a seller do in the Bayesian setting where the daily valuations are drawn from a known distribution? Taking the expected valuation for each day is, of course, not a valid solution (the pricing for an agent that has a value of $2$ or $0$ is very different than for agent with value $1$). This problem can be seen as a type of Bayesian Stackelberg game (with each set of valuations considered as a type of ``follower''). However, in general, finding the optimal strategy in such games is known to be be NP-hard~\cite{CS06}. But our particular structure (with a known distribution for each day's value), may allow for better results.

\subsection*{Acknowledgements}
This work was partially supported by NSERC Discovery Grant 482671 and NSERC Accelerator Grant Fund 503949.
%%%%%%%%%%%%%%%%%%%%%%%%%%%%%%%%%%%%%%%%%%%%%%%%%%%%%%%%%%%%%%%%%%%%%%%%%%%%%%%%%%%%%%%%%%%%%%%%%%%%%%%%%
%% bibliography: see CFP for number of permitted pages

\bibliographystyle{ACM-Reference-Format} % do not change this line!
\bibliography{bib} % put name of your .bib file here

\end{document}